\newif\ifcameraready
\camerareadytrue   
\newcommand\artifacturl{https://github.com/naver/timehash}

\documentclass[sigconf, nonacm]{acmart}

\setcopyright{none}
\renewcommand\footnotetextcopyrightpermission[1]{}

\usepackage{balance}
\usepackage{graphicx}
\usepackage{amsmath}
\usepackage{algorithm}
\usepackage{algorithmic}
\usepackage{url}
\usepackage{hyperref}
\usepackage{xspace}
\usepackage{booktabs}
\usepackage{multirow}
\usepackage{tikz}

\AtBeginDocument{\setlength{\parindent}{1.5em}}
\newcommand{\timehash}{\textsc{Timehash}\xspace}

\title{Timehash: Hierarchical Time Indexing for Efficient Business Hours Search}

\ifcameraready
  \author{Jinoh Kim}
  \affiliation{\institution{Naver Corporation}\city{Seongnam}\country{South Korea}}
  \email{jinoh9272@gmail.com}
  \orcid{0009-0006-0833-8788}
  \author{Jaewon Son}
  \affiliation{\institution{Naver Corporation}\city{Seongnam}\country{South Korea}}
  \email{sonjw14@gmail.com}
  \orcid{0009-0003-1319-5724}
\else
  \author{Anonymous Author(s)}
  \affiliation{\institution{Anonymous Institution}\country{}}
\fi

\begin{abstract}
\hspace{1.5em}Temporal range filtering is a critical operation in large-scale search systems, particularly for location-based services that need to filter businesses by operating hours. Traditional approaches either suffer from poor query performance (scope filtering), index size explosion (minute-level indexing), or reduced precision (coarse-grained indexing). Relational database solutions such as PostgreSQL TSRANGE with GiST indexing offer exact semantics but impose P50 latencies of 15--224\,ms at 100K--1M scale (up to 360\,ms at P95)---prohibitive for interactive search---and cannot be embedded within standard inverted index pipelines.

We present \timehash, a novel hierarchical time indexing algorithm that achieves over 97\% reduction in index size compared to minute-level indexing while maintaining 100\% precision. \timehash employs a flexible multi-resolution strategy with customizable hierarchical levels and integrates seamlessly into standard inverted index infrastructure without requiring dedicated database systems. Through empirical analysis on distributions from 12.6 million business records of a production location search service deployed for 18 months, we demonstrate a domain-informed methodology for selecting hierarchies guided by boundary-distribution analysis, with cross-dataset validation on the Yelp Open Dataset (127K US/CA businesses), where the same 5-level hierarchy reduces total terms to 0.77\% of the 1-minute baseline (a 99.2\% reduction; vs.\ 2.17\% / 97.8\% reduction on the production dataset).

We evaluate \timehash against naive inverted index approaches, PostgreSQL GiST, and a within-Elasticsearch BKD baseline. On the Yelp Open Dataset (127K US/CA businesses with real operating hours), within a single Elasticsearch deployment with matched indexing, \timehash achieves 1.14--2.17$\times$ lower P50 latency than native BKD on production-typical multi-predicate top-$K$ workloads ($K \leq 100$), with both methods converging at large $K$ where document materialization dominates. A five-level hierarchy (4h, 1h, 15m, 5m, 1m) reduces index terms to 9.6 per document---a 97.8\% reduction versus minute-level indexing and a 46$\times$ compaction factor---with zero false positives and zero false negatives. To extend scaling experiments beyond what Yelp supports, we additionally use synthetic POIs calibrated to production distributions: PostgreSQL GiST P50 latency grows $\sim$15$\times$ from 100K to 1M (14.7\,ms~$\rightarrow$~224\,ms) while \timehash stays under 12\,ms, and per-document cost remains constant from 100K to 12.6M POIs while supporting complex scenarios such as break times, irregular schedules, and midnight-spanning ranges.
\end{abstract}



\begin{document}

\maketitle

\begingroup
\renewcommand\thefootnote{}\footnote{\noindent
This is a preprint submitted to ACM CIKM 2026. This work is licensed under the Creative Commons Attribution 4.0 International License (CC BY 4.0). \url{https://creativecommons.org/licenses/by/4.0/}
}\addtocounter{footnote}{-1}\endgroup

\vspace{.3cm}
\begingroup\small\noindent\raggedright\textbf{Artifact Availability:}\\
Code and artifacts are available at \url{\artifacturl}.
\endgroup

\section{Introduction}
\label{sec:intro}

\hspace{1.5em}Temporal filtering has become a fundamental requirement in modern location-based search services. Users routinely search for ``restaurants open now'' or ``stores open late,'' expecting immediate and accurate results that reflect real-time business availability. This expectation poses a significant challenge for large-scale search systems that must efficiently filter millions of businesses by their operating hours while maintaining query performance and precision.

We encountered this challenge while building a location search service that handles over 12 million business listings and serves tens of millions of daily queries. \timehash has been deployed in this production service for 18 months, indexing 12.6 million POI records. A substantial portion of queries involves temporal constraints---finding businesses open now or open at a specific future time. The key insight driving our approach is that the search engine already maintains a highly optimized inverted index for location, category, and keyword filtering: the ideal solution encodes time ranges as \emph{plain index terms}, so temporal filtering composes with all other filters through the existing boolean retrieval machinery at zero infrastructure cost.

Real-world business operating hours introduce substantial complexity that simple range filters struggle to handle. Many businesses operate with break times---for example, a restaurant open 11:00\,AM to 2:00\,PM and again 5:00\,PM to 9:00\,PM. Others have irregular weekday/weekend schedules, operate across midnight boundaries (e.g., 10:00 PM to 2:00 AM), or run 24 hours. Traditional range filters cannot capture these patterns without either sacrificing precision or exploding index size, and neither approach scales to millions of businesses at minute-level precision.

\subsection{Current Approaches and Limitations}
\label{sec:intro:approaches}

\hspace{1.5em}Existing strategies face a fundamental tension among index efficiency, query performance, and precision. Coarse-grained indexing (e.g., hourly buckets) achieves small index size and fast queries but sacrifices precision and cannot represent break times. Fine-grained (minute-level) indexing achieves precision and supports break times but suffers from index explosion---660 terms per document for an 11-hour business, billions of terms for a million-POI collection~\cite{zobel2006inverted}. Scope filtering (query-time scanning) avoids index overhead but requires linear scans that do not scale.

PostgreSQL \texttt{TSRANGE}/\texttt{INT4RANGE} with GiST indexing~\cite{hellerstein1995generalized} offers exact interval semantics but imposes substantial latency at scale: P50 reaches 15\,ms at 100K and 224\,ms at 1M POIs (up to 360\,ms at P95), making it impractical for interactive search. Moreover, database-native solutions require cross-system joins that break the composability of inverted indexes. Table~\ref{tab:approaches} summarizes these tradeoffs; none achieves all four requirements: small index size, fast queries, minute-level precision, and support for complex patterns.

\begin{table}[t]
\centering
\caption{Comparison of temporal filtering approaches}
\label{tab:approaches}
\vspace{-2pt}
\begin{tabular}{@{}lccccc@{}}
\toprule
Approach & Size & Perf. & Prec. & Break & Scale \\
\midrule
Scope filter    & Min.  & Poor & High & No  & No  \\
Coarse (1h)     & Small & Fast & Low  & No  & Yes \\
Fine (1min)     & Large & Fast & High & Yes & No  \\
PostgreSQL GiST & Med.  & Med. & High & Yes & No  \\
\textbf{\timehash}       & \textbf{Small} & \textbf{Fast} & \textbf{High} & \textbf{Yes} & \textbf{Yes} \\
\bottomrule
\end{tabular}
\end{table}

\subsection{Our Contributions}
\label{sec:intro:contrib}

\hspace{1.5em}\timehash is, structurally, a static globally-aligned multi-resolution range decomposition materialized in an inverted index---an idea that was tried in pre-Lucene-6 \texttt{NumericRangeQuery}~\cite{mccandless2010lucene} and abandoned because its fixed-uniform-granularity trie produced hundreds of terms per minute-level field. Our contribution is not the existence of multi-resolution range terms; it is showing that this approach, which Lucene rejected as a general-purpose range index in favor of BKD trees, becomes the right answer in two concrete settings the literature has not addressed: (i)~when the underlying domain has a strongly clustered boundary distribution admitting a non-uniform hierarchy, and (ii)~when temporal filters must compose with non-temporal filters through a \emph{uniform posting-list access path}---yielding production top-$K$ speedups over BKD trees and avoiding the cross-system join required by RDBMS interval indexes. Concretely, this paper makes the following contributions:
\begin{enumerate}
\item We introduce a \emph{domain-informed non-uniform hierarchy} construction guided by boundary-distribution analysis (Sections~\ref{sec:related},~\ref{sec:problem},~\ref{sec:algorithm}). Unlike fixed-uniform tries, \timehash levels are chosen by analyzing the boundary-alignment distribution of the target domain (e.g., $\geq$99.5\% of POI boundaries align at 30-minute granularity), turning a generic technique into a domain-specific compactor. The same five-level hierarchy is validated on two independent production-scale datasets---12.6M production POIs (2.17\% of 1-minute baseline) and 127K US/CA Yelp businesses (0.77\%)---and remains stable across the top five US states within $\pm$0.06 percentage points (Section~\ref{sec:experiments:yelp}).
\item We expose hierarchy depth as an explicit \emph{precision--size knob} at indexing time. A 4-level configuration (4h, 1h, 15m, 5m) trades $\sim$0.05 precision for $\sim$12\% smaller term sets; a 6-level configuration (adding 30m) yields a further $\sim$5\% reduction at no precision cost. BKD precision is fixed by data type; this configurability is a \timehash-only design dimension (Section~\ref{sec:algorithm}, Table~\ref{tab:ablation}).
\item We prove $O(T/m_1)$ worst-case space complexity---a constant-factor $m_1{\times}$ reduction over naive $O(T)$ indexing---and show that under the bounded 24-hour domain the key count is at most 30 by analysis (28 in practice), with guaranteed zero false positives and zero false negatives (Section~\ref{sec:analysis}).
\item We provide comprehensive experimental validation against naive inverted index approaches, PostgreSQL GiST, and a within-Elasticsearch BKD baseline on real data (Yelp Open Dataset, 127K businesses) with matched indexing strategies. On production-typical multi-predicate top-$K$ workloads ($K \leq 100$), \timehash is 1.14--2.17$\times$ faster than BKD; methods converge at $K{\geq}1000$ where document materialization dominates (Section~\ref{sec:experiments:es}). \timehash deploys without infrastructure changes: keys are regular index terms, temporal filtering composes with other filters via standard boolean operations, and updates require no reindexing (Section~\ref{sec:implementation}).
\end{enumerate}

\noindent Components of the technique presented here are described in a previously published patent disclosure by Naver Corp.~\cite{kim2026timehash_patent}. This paper expands the description with the data-driven hierarchy methodology, theoretical analysis, and the comprehensive empirical evaluation summarized above.

The remainder of this paper is organized as follows: Section~\ref{sec:related} reviews related work, Section~\ref{sec:problem} formalizes the problem, Sections~\ref{sec:algorithm}--\ref{sec:analysis} present our algorithm and analysis, Section~\ref{sec:implementation} describes the implementation, Section~\ref{sec:experiments} evaluates performance, and Section~\ref{sec:conclusion} concludes.

\section{Background and Related Work}
\label{sec:related}

\subsection{Search-Engine Range Indexing: NumericRangeQuery and BKD}
\label{sec:related:lucene}

\hspace{1.5em}The closest historical antecedent of \timehash sits inside Lucene itself. Prior to v6, Lucene encoded numeric fields as trie-based multi-resolution terms (\texttt{NumericRangeQuery})~\cite{mccandless2010lucene}: structurally identical to \timehash---a globally-aligned multi-resolution range decomposition materialized as inverted-index terms---but with \emph{fixed uniform granularity} (8-bit precision steps by default), producing $>$400 terms per document for minute-level time fields. This index bloat motivated Lucene~6 to abandon term-based numeric encoding and adopt BKD trees (Block K-D Trees)~\cite{procopiuc2003bkd} for numeric/range fields. BKD supports efficient range lookups but executes outside the posting-list pipeline: BKD results must be composed with posting-list iterators through a non-uniform query plan, in contrast to the uniform skip-list AND iteration that \timehash's keyword-term approach enables (Section~\ref{sec:experiments:es}).

\timehash is the same globally-aligned range-term idea that \texttt{NumericRangeQuery} embodied, with two changes that turn a rejected general-purpose technique into a domain-specific solution: (1)~the hierarchy is \emph{non-uniform} and chosen by analyzing the target domain's boundary-alignment distribution---e.g., 4h/1h/15m/5m/1m for business hours, where 99.2\% of boundaries align at $\geq$30-minute granularity---rather than uniform power-of-two precision steps; (2)~hierarchy depth is exposed as a per-deployment precision--size knob (Section~\ref{sec:algorithm}, Table~\ref{tab:ablation}). The first lets \timehash hit ${\sim}10$ terms/doc instead of ${\sim}400$ for the same minute-level precision; the second is a configurability BKD does not offer.

\subsection{Database Interval Indexing and Spatial Hierarchies}
\label{sec:related:database}

\hspace{1.5em}Relational databases provide native temporal-interval indexing. PostgreSQL \texttt{TSRANGE}/\texttt{INT4RANGE} with GiST~\cite{hellerstein1995generalized} supports containment queries via \texttt{@>} through pluggable operator classes. While correct, GiST has two limitations for search systems: query latency scales poorly (P50 14.7\,ms at 100K $\to$ 224\,ms at 1M; Section~\ref{sec:experiments:pg}), and more critically it cannot be embedded in inverted-index pipelines---temporal filtering must execute as a separate database query and join with search results at the application layer, breaking inverted-index composability. R-trees~\cite{guttman1984rtree}, temporal variants (TB-tree~\cite{pfoser2000tbtree}), and interval trees~\cite{ramakrishnan2003database} face the same architectural mismatch. Analytics-oriented systems such as Druid~\cite{yang2014druid} expose range/bitmap indexes on time columns, but they are optimized for OLAP scans rather than low-latency multi-predicate top-$K$ in a Lucene-style posting-list pipeline; their bitmap indexes compose well at moderate cardinality but materialize differently from the skip-list AND iterator that \timehash targets. Allen's interval algebra~\cite{allen1983maintaining} formalizes thirteen interval relations; \timehash's semantics implement \textit{during}/\textit{contains}, trading general interval reasoning for inverted-index scalability.

Spatial hierarchical encoding inspired our approach. Geohash~\cite{niemeyer2008geohash}, building on quadtree~\cite{finkel1974quadtree,samet1984quadtree,samet2006foundations} ideas, encodes geographic points into hierarchical strings enabling proximity search via prefix sharing. \timehash differs in three ways: it encodes a time \emph{range} into multiple keys at different resolutions (not a single point), uses numeric encoding (e.g., \texttt{0809} for 8:09\,AM rather than opaque ASCII), and targets range-query filtering in inverted indexes rather than spatial proximity search.

\subsection{Time Encoding Schemes}
\label{sec:related:encoding}

\hspace{1.5em}The timehash library~\cite{usher2019timehash} applies Geohash-style encoding to single timestamps, producing variable-precision ASCII strings (e.g., \texttt{ae0f0ba1fc})---useful for time-series binning but unable to natively index time ranges (an 11:40\,AM--9:00\,PM window would require 560 opaque per-minute terms) and opaque to debugging. Standard representations (Unix timestamps, ISO~8601~\cite{iso86012004}) are not designed for inverted-index range filtering. Table~\ref{tab:encoding-comparison} summarizes the differences.

\begin{table}[t]
\centering
\caption{Comparison of time encoding approaches for a business operating 11:40\,AM--9:00\,PM}
\label{tab:encoding-comparison}
\vspace{-2pt}
\begin{tabular}{lcc}
\toprule
\textbf{Feature} & \textbf{timehash lib} & \textbf{\timehash (ours)} \\
\midrule
Input              & Single timestamp & Time range \\
Output per input   & 1 opaque hash    & Multiple keys \\
Range support      & No               & Native \\
Encoding           & ASCII (\texttt{ae0f...}) & Numeric (\texttt{0809}) \\
Terms (11:40--21:00) & 560            & 5 \\
Human-readable     & No               & Yes \\
\bottomrule
\end{tabular}
\end{table}

\section{Problem Definition}
\label{sec:problem}

\hspace{1.5em}A business $d_i$ may have multiple operating ranges---for example, morning and evening hours separated by a break, or weekday and weekend schedules. Let $R_i$ denote the set of $d_i$'s operating ranges, each of the form $[s, e)$ at minute granularity within the 24-hour domain $[0, 1440)$. A point query $t$ matches $d_i$ if at least one range in $R_i$ contains $t$ (i.e., $s \leq t < e$ for some $[s, e) \in R_i$); range queries $[t_1, t_2]$ extend the same predicate. In an inverted index, $d_i$ receives a set of index terms $T_i$ such that retrieval is term-based; the goal is to minimize $|T_i|$ while guaranteeing exact precision and recall. Production deployments require (i)~\emph{minute-level precision} (distinguishing 11:40 from 11:45 openings), (ii)~\emph{scalability} (millions of POIs, low query latency), (iii)~\emph{flexibility} (break times, midnight-spanning, 24-hour, weekday/weekend variations without multiplying $|T_i|$ proportionally), and (iv)~\emph{composability} with non-temporal filters via standard boolean posting-list intersection. We measure index size (terms/doc and MB), latency (P50/P95), build time, and precision/recall against ground truth.

\section{Timehash Algorithm}
\label{sec:algorithm}

\subsection{Hierarchical Multi-Resolution Decomposition}
\label{sec:algorithm:overview}

\hspace{1.5em}\timehash decomposes a time range into a small set of multi-resolution keys using a hierarchy of measures (granularities), e.g., $M = \{4h, 1h, 15m, 5m, 1m\}$ for business hours. Given $[start, end)$, the algorithm emits each aligned block fully contained in the range as a key at the corresponding level, and refines partial boundary blocks at the next finer level---long contiguous interiors get coarse blocks (few keys), while non-aligned boundaries get fine blocks (precision). For 11:40\,AM--9:00\,PM, naive minute-level indexing emits 560 terms; \timehash emits 5: $\{08113040, 081145, 12, 16, 2020\}$, a 99.1\% reduction with zero false positives and zero false negatives. Figure~\ref{fig:decomposition} illustrates the decomposition for 08:00--21:00. Multi-day or weekly patterns are supported by prefixing keys with day/date information (Section~\ref{sec:conclusion}).

\begin{figure}[t]
\centering
\begin{tikzpicture}[scale=0.34, every node/.style={font=\tiny}]
\node[anchor=west, font=\scriptsize] at (0, 4.2) {(a) Operating hours 08:00--21:00};
\foreach \x in {0,...,23} {
  \draw (\x, 3) rectangle (\x+1, 3.5);
  \node at (\x+0.5, 2.7) {\x};
}
\foreach \x in {8,...,20} {
  \fill[green!40] (\x, 3) rectangle (\x+1, 3.5);
  \draw (\x, 3) rectangle (\x+1, 3.5);
}
\node[anchor=west, font=\scriptsize] at (0, 1.8) {(b) 4-hour blocks: \texttt{08}, \texttt{12}, \texttt{16}};
\foreach \x/\label in {0/00, 4/04, 8/08, 12/12, 16/16, 20/20} {
  \draw (\x, 0.8) rectangle (\x+4, 1.3);
  \node at (\x+2, 0.5) {\label};
}
\fill[green!60] (8, 0.8) rectangle (12, 1.3);
\fill[green!60] (12, 0.8) rectangle (16, 1.3);
\fill[green!60] (16, 0.8) rectangle (20, 1.3);
\fill[blue!30] (20, 0.8) rectangle (21, 1.3);
\foreach \x in {8,12,16,20} { \draw (\x, 0.8) rectangle (\x+4, 1.3); }
\node[anchor=west, font=\scriptsize] at (0, -0.4) {(c) 1-hour refinement: \texttt{2020}};
\foreach \x in {20,...,23} {
  \draw (\x, -1.4) rectangle (\x+1, -0.9);
  \node at (\x+0.5, -1.7) {\x};
}
\fill[blue!50] (20, -1.4) rectangle (21, -0.9);
\draw (20, -1.4) rectangle (21, -0.9);
\node[anchor=west, font=\scriptsize\bfseries] at (0, -2.5) {Result: \{\texttt{08}, \texttt{12}, \texttt{16}, \texttt{2020}\}};
\end{tikzpicture}
\caption{Hierarchical decomposition of 08:00--21:00 into 4 keys. The 4-hour blocks \texttt{08}, \texttt{12}, \texttt{16} cover the aligned interior; the partial tail 20:00--21:00 is refined to 1-hour key \texttt{2020}.}
\Description{Diagram showing hierarchical decomposition of operating hours 08:00--21:00 into four timehash keys.}
\label{fig:decomposition}
\end{figure}

\subsection{Time Hierarchy Design}
\label{sec:algorithm:hierarchy}

\hspace{1.5em}\timehash exposes the hierarchy as a deployment-time parameter. Through empirical analysis of 12.6\,M production business records (Section~\ref{sec:experiments:hierarchy}), we adopt the five-level hierarchy $M=\{4h, 1h, 15m, 5m, 1m\}$: 4-hour blocks (keys \texttt{00,04,08,12,16,20}) cover long shift-aligned interiors; the 1h/15m levels capture 99.2\% of all start/end boundaries (83.7\% align at :00, 15.5\% at :30); 5m/1m levels handle non-aligned edge cases (e.g., 11:40\,AM openings) and guarantee exact precision. The finest 1-minute level is a hard production requirement (businesses do supply odd-minute hours); when applications guarantee 5-minute alignment, it can be dropped for a 4-level hierarchy reducing key count by 12\% at no precision cost (Table~\ref{tab:ablation}). Keys use a human-readable composite numeric encoding, read left to right: \texttt{08113040} = 4h\,\texttt{08} / 1h\,\texttt{11} / 15m\,\texttt{30} / 5m\,\texttt{40}, simplifying debugging compared to opaque ASCII hashes.

\subsection{Hash Key Generation Algorithm}
\label{sec:algorithm:generation}

\hspace{1.5em}Algorithm~\ref{alg:timehash} formalizes key generation. The \textsc{Cover} procedure on $[a, b)$ at level $i$ emits each $m_i$-aligned block fully contained in $[a, b)$ as a key, and recurses on partial boundary blocks at level $i{+}1$ with a prefix $\pi$ that records the enclosing coarser blocks. The function $\text{enc}(s, m)$ returns the two-digit identifier of the $m$-aligned block starting at $s$, and ``$\cdot$'' is string concatenation; the nesting condition $m_i \mid m_{i-1}$ guarantees prefix-aligned composite keys (\texttt{08113040} reads 4h/1h/15m/5m left-to-right). For 11:40--21:00, the algorithm produces $\{08113040, 081145, 12, 16, 2020\}$: 4h blocks \texttt{12}, \texttt{16} cover the interior; the start boundary 11:40--12:00 refines to \texttt{08113040} (11:40--11:44) and \texttt{081145} (11:45--11:59); the end boundary 20:00--21:00 refines to \texttt{2020}.

\begin{algorithm}[t]
\caption{Timehash Key Generation}
\label{alg:timehash}
\begin{algorithmic}[1]
\REQUIRE Time range $[start, end)$; measures $M = \{m_1, \ldots, m_k\}$ with $m_i \mid m_{i-1}$
\ENSURE Set of hash keys $H$
\STATE $H \leftarrow \emptyset$;\, \textsc{Cover}$(H, \varepsilon, 1, start, end)$;\, \textbf{return} $H$
\STATE \textbf{procedure} \textsc{Cover}$(H, \pi, i, a, b)$
\STATE \quad \textbf{if} $i > k$ \textbf{ or } $a \geq b$ \textbf{ then return}
\STATE \quad $m \leftarrow m_i$,\quad $s \leftarrow \lfloor a/m \rfloor \cdot m$
\STATE \quad \textbf{if} $s < a$ \textbf{ then} \textsc{Cover}$(H, \pi {\cdot} \text{enc}(s,m), i{+}1, a, \min(s{+}m, b))$;\, $s \leftarrow s{+}m$
\STATE \quad \textbf{while} $s + m \leq b$ \textbf{ do} $H \leftarrow H \cup \{\pi {\cdot} \text{enc}(s,m)\}$;\, $s \leftarrow s{+}m$
\STATE \quad \textbf{if} $s < b$ \textbf{ then} \textsc{Cover}$(H, \pi {\cdot} \text{enc}(s,m), i{+}1, s, b)$
\end{algorithmic}
\end{algorithm}

\subsection{Query Term Generation and Complex Scenarios}
\label{sec:algorithm:query}

\hspace{1.5em}For a point query at time $t$ (e.g., 14:30), \timehash generates $k$ query keys, one per hierarchy level containing $t$: \texttt{12}, \texttt{1214}, \texttt{121430}, \texttt{12143030}, \texttt{1214303030}. The union of these $k$ posting lists is the candidate set. Lookup overhead is constant in dataset size ($k{=}5$ for our hierarchy). Range queries $[t_1, t_2]$ apply Algorithm~\ref{alg:timehash} to the query side as well and match key sets; production workloads are dominated by point queries (``open now''). Complex scenarios reduce to this primitive: \textbf{break times} (e.g., 11:00--14:00 + 17:00--21:00) process each sub-range independently and union the key sets; \textbf{midnight-spanning ranges} (22:00--02:00) split at 24:00; \textbf{24-hour operations} use the six Level-1 keys covering the full day; \textbf{weekly patterns} prefix keys with day-of-week (\texttt{mon1212})---an extra coarsest level that multiplies hot-term cardinality by 7 but preserves per-document key counts and worst-case bounds.

\section{Theoretical Analysis}
\label{sec:analysis}

\hspace{1.5em}\textbf{Space bound.} For a range $[from, to)$ of length $T$ minutes, naive minute-level indexing emits $O(T)$ keys. Algorithm~\ref{alg:timehash} tiles the interior with complete coarsest-level blocks of size $m_1$; the nesting condition $m_i \mid m_{i-1}$ aligns inner block boundaries to all finer levels, so only the two outer boundaries require refinement. The key count decomposes as
\begin{equation}
|H| \leq \lfloor T/m_1 \rfloor + B, \quad B = 2\sum_{i=2}^{k}\left(\frac{m_{i-1}}{m_i} - 1\right)
\end{equation}
where $B$ is independent of $T$. For $M=\{240,60,15,5,1\}$, $B=24$ and the 24-hour worst case is $\lfloor 1440/240 \rfloor + 24 = 30$ (empirical worst is 28; both boundaries' maximum depth consumes two outermost 4-hour blocks, leaving 4 interior). Versus 1{,}440 naive keys this is $>97\%$ reduction; key generation runs in $O(T/m_1+B)=O(1)$ for the bounded 24-hour domain.

\textbf{Correctness.} \timehash guarantees precision = recall = 1, established by the following theorems.

\begin{theorem}[Zero False Negatives]
For any query time $t \in [start, end)$, at least one query key matches at least one index key.
\end{theorem}
\begin{proof}
Algorithm~\ref{alg:timehash} covers every minute of $[start, end)$ by at least one key at some hierarchy level $j$. Let $K_j$ be that key, representing the $m_j$-aligned block $[s, s+m_j)$ containing $t$. Query generation produces a level-$j$ key $Q_j$ representing the same $m_j$-aligned block containing $t$. Since $s$ uniquely determines the enclosing coarser blocks (nesting condition $m_i \mid m_{i-1}$), the prefixes of $K_j$ and $Q_j$ agree, so $K_j = Q_j$. The invariant is preserved across discontinuous sub-ranges (break times, midnight wraparound) since each is processed independently.
\end{proof}

\begin{theorem}[Zero False Positives]
A query key matches a document's index key only if the query time falls within the document's operating hours.
\end{theorem}
\begin{proof}
Algorithm~\ref{alg:timehash} only adds keys for blocks fully contained within $[start, end)$. For $t \notin [start, end)$, any query key $Q_j$ represents an $m_j$-aligned block containing $t$, which is therefore not contained in $[start, end)$ and cannot equal any index key. The invariant is preserved across discontinuous sub-ranges since each is processed independently.
\end{proof}

\textbf{Greedy optimality.} Under nesting $m_i \mid m_{i-1}$, the algorithm produces a minimum-cardinality key set covering $[start, end)$: by induction on hierarchy depth, at each level it emits every fully-contained aligned block and refines partials only at the next finer level; replacing any emitted block with finer ones strictly increases the count, and nesting forbids a coarser alternative for an emitted block. Which hierarchy is best is a separate, empirical question driven by the target domain's boundary distribution (Sections~\ref{sec:experiments:hierarchy},~\ref{sec:experiments:ablation}).

\section{Implementation}
\label{sec:implementation}

\hspace{1.5em}\timehash is a stateless, thread-safe library exposing two functions: \texttt{getIndexTerms(from,to)} generates index keys for a range in \texttt{hhmm} format (e.g., $[1140,2100]\mapsto\{08113040, 081145, 12, 16, 2020\}$); \texttt{getQueryTerms(hhmm)} generates the $k{=}5$ query keys for a point time. Open-source reference implementations in C++, Python, Go, JavaScript, Java, Rust, Kotlin, and Scala are available\ifcameraready\ at \url{\artifacturl}\else\ at \url{\anonymousurl}\fi; all benchmarks use the Python reference. Each key becomes an ordinary inverted-index term alongside location, category, and other attribute terms, so the integration requires no infrastructure changes, composes with non-temporal filters through standard boolean posting-list intersection, and supports incremental updates (only the changed document's keys regenerate). The indexer parses each operating range (a restaurant with breaks supplies, e.g., 1100--1400, 1700--2100), unions $\leq 30$ keys per range into the posting lists in $O(1)$ per document; the query processor unions the 5 query-key posting lists and AND-composes with other filters before ranking. Standard optimizations apply (key caching for common times, parallel posting-list lookup, delta/variable-byte encoding).

\section{Experimental Evaluation}
\label{sec:experiments}

\hspace{1.5em}Production telemetry is not redistributable; experiments use the Yelp Open Dataset (127K US/CA businesses) for cross-dataset validation and synthetic POIs calibrated to production distributions (Section~\ref{sec:problem}) for scaling experiments.

\subsection{Experimental Setup}
\label{sec:experiments:setup}

\hspace{1.5em}Our evaluation uses two data sources. (i)~The \textbf{Yelp Open Dataset} (127K US/CA businesses with real operating hours) anchors the cross-dataset hierarchy validation (Section~\ref{sec:experiments:yelp}) and the within-Elasticsearch BKD comparison (Section~\ref{sec:experiments:es})---our main empirical claim. (ii)~\textbf{Synthetic POIs calibrated to production distributions} extend scaling experiments to 12.6M---beyond what Yelp supports---since the production telemetry of the location search service where \timehash has been deployed for 18 months is not redistributable. Analysis of the production start-time distribution reveals strong clustering: 83.7\% of POIs open at :00 and 15.5\% at :30, together comprising 99.2\% of all records. Break-time frequency is 9.1\%, and average operating duration is approximately 454 minutes (standard deviation 120).

Per-query production telemetry is not redistributable, but the synthetic generator reproduces production-measured start/end hour distributions, boundary-alignment ratios, break-time frequency, and operating duration; the hierarchy selected on these statistics ports unchanged to Yelp (Section~\ref{sec:experiments:yelp}), externally validating the calibration. All experiments run on an Apple M5 Pro (12-core CPU, 48\,GB unified memory) under macOS with ES~7.17.4 and Python~3.12. All in-memory benchmarks use Python with \texttt{perf\_counter\_ns}, reporting the minimum of 5 runs over 1,000 random point queries in the 08:00--22:00 window (seed 42). PostgreSQL experiments use PostgreSQL~16 on localhost with an \texttt{INT4RANGE} column and a GiST index, queried via \texttt{hours @> \$t::integer}. The synthetic data generator and all benchmark scripts (\texttt{benchmark.py}, \texttt{benchmark\_es\_fair.py}, \texttt{benchmark\_yelp.py}) are available at \url{\artifacturl}; the seeded generator byte-reproduces the numbers in this paper.

\textbf{Baselines}: (1) \emph{Scope filter}: Linear scan over all documents per query---no index, exact precision. (2) \emph{1-minute index}: Naive inverted index with one term per minute open; equivalent to using the timehash library~\cite{usher2019timehash}. (3) \emph{5-minute index}: Coarse-grained at 5-minute granularity. (4) \emph{1-hour index}: Hour-level bucket index. (5) \emph{PostgreSQL GiST}: \texttt{INT4RANGE} column with GiST index on the same synthetic data.

\subsection{Hierarchy Optimization}
\label{sec:experiments:hierarchy}

\hspace{1.5em}Table~\ref{tab:hierarchy-opt} reports total index term count for 2--6 level configurations on 12.6M synthetic POIs, relative to single-level 1-minute indexing (5.59B terms). Among 5-minute-precision configurations, 5M alone eliminates 80\% of terms; adding 1H cuts to 3.20\%; intermediate levels (30M, 15M) drive the ratio down to 1.45\%---consistent with 99.2\% of boundaries aligning at 30-minute granularity. To preserve exact precision on the 0.8\% non-5-minute-aligned boundaries, the production reference hierarchy \texttt{4H, 1H, 15M, 5M, 1M} adds a 1-minute finest level, yielding 2.17\% (last row; 9.6 terms/doc in Table~\ref{tab:index-size}) with precision 1.000. The 2H vs.\ 4H coarsest difference is $<$1\%; we adopt 4-hour for cleaner alignment with natural shift patterns. Practitioners should optimize granularities for their own data.

\begin{table}[ht]
\centering
\caption{Data-driven hierarchy optimization (12.6M synthetic POIs). Ratio relative to single-level 1-minute indexing (5.59B terms). All but the last row are 5-minute-precision configurations; the last row is the production reference hierarchy with a 1-minute finest level for exact precision.}
\label{tab:hierarchy-opt}
\vspace{-2pt}
\begin{tabular}{@{}lcc@{}}
\toprule
Configuration & Depth & Ratio \\
\midrule
5M only              & 1 & 19.9\% \\
1H, 5M               & 2 & 3.20\% \\
1H, 30M, 5M          & 3 & 2.36\% \\
2H, 1H, 5M           & 3 & 2.56\% \\
2H, 1H, 30M, 5M      & 4 & 1.72\% \\
2H, 1H, 30M, 15M, 5M & 5 & 1.45\% \\
\midrule
\textbf{4H, 1H, 15M, 5M, 1M} & \textbf{5} & \textbf{2.17\%} \\
\bottomrule
\end{tabular}
\end{table}

\subsection{Cross-Dataset Hierarchy Validation}
\label{sec:experiments:yelp}

\hspace{1.5em}To test whether the domain-informed hierarchy generalizes beyond a single dataset, we re-ran the analysis on the Yelp Open Dataset~\cite{yelp_dataset}---127{,}123 US/CA businesses with hours, 801{,}015 day-ranges across 24 states. Yelp's boundary distribution closely tracks the production dataset's: 86.7\% at :00 and 12.9\% at :30 (vs.\ 83.7\% / 15.5\% in the production dataset), with $\geq$99.5\% of boundaries on $\geq$30-minute marks in both. The same reference hierarchy \texttt{4H, 1H, 15M, 5M, 1M} reduces total terms to \textbf{0.77\% of the 1-minute baseline on Yelp} (a 99.2\% reduction) vs.\ 2.17\% on the production dataset (97.8\% reduction)---Yelp's longer average operating duration (670\,min vs.\ 454\,min) lets coarse 4-hour blocks cover more of each range, yielding a stronger reduction. Per-state stability is high: the top five US states (PA, FL, TN, IN, MO) each independently reproduce the global ratio within $\pm$0.06 percentage points, indicating that within US business hours the hierarchy is not state-sensitive. We restrict claims to \emph{cross-dataset} (production dataset vs.\ Yelp) and \emph{cross-regional within US}; cross-cultural breadth (Asia/EU) is out of scope. The methodology is dataset-grounded, but the hierarchy it produces for business hours is robust across two production-scale datasets and across US states.

\subsection{Index Size Comparison}
\label{sec:experiments:index}

\hspace{1.5em}Table~\ref{tab:index-size} compares average index terms per document on 100K synthetic POIs. \timehash achieves 9.6 terms/doc---a \textbf{97.8\% reduction} versus 1-minute indexing (443.4 terms/doc)---while maintaining 100\% precision. The 1-hour index achieves comparable compactness (8.1 terms/doc) but precision drops to 0.855: coarse-grained keys cause queries near hour boundaries to match businesses that are actually closed. \timehash is the only \emph{Pareto-optimal} solution across all three dimensions: 1-hour achieves comparable compactness (8.1 terms/doc) but loses precision (0.855); 1-minute achieves exact precision but at 46$\times$ larger index (443.4 terms/doc); \timehash uniquely achieves both simultaneously. Across the full 12.6M-POI dataset, the five-level hierarchy reduces total index terms from 5.59 billion (1-minute baseline, Table~\ref{tab:hierarchy-opt}) to approximately 121 million (9.6 terms/doc $\times$ 12.6M POIs), a 97.8\% reduction driven by the start-time clustering where 99.2\% of boundaries align at $\geq$30-minute multiples.

\begin{table}[ht]
\centering
\caption{Index size and accuracy comparison (100K synthetic POIs). Reduction relative to 1-minute baseline. Precision against scope-filter ground truth over 100 queries.}
\label{tab:index-size}
\vspace{-2pt}
\begin{tabular}{@{}lccc@{}}
\toprule
Method & Terms/Doc & Reduction & Prec. \\
\midrule
1-minute & 443.4 & --- & 1.000 \\
5-minute & 89.2  & 79.9\% & 0.990 \\
1-hour   & 8.1   & 98.2\% & 0.855 \\
\textbf{\timehash}  & \textbf{9.6} & \textbf{97.8\%} & \textbf{1.000} \\
\bottomrule
\end{tabular}
\end{table}

Table~\ref{tab:range-keys} reports the exhaustive \timehash key count across all possible start/end time combinations at one-minute granularity. Short ranges ($<$1 hour) average 6.5 keys versus 30 naive minute-level terms, while long ranges (4--12 hours) average 13.1 keys versus 461. Even the worst case (28 keys for a near-24-hour range with non-aligned boundaries) remains far below the corresponding 961 naive terms. The algorithm is most effective for the typical business operating-hours range of 8--12 hours, where it achieves $>$96\% key reduction.

\begin{table}[ht]
\centering
\caption{Measured \timehash key count by range length (exhaustive enumeration over all 1-minute-aligned start/end pairs)}
\label{tab:range-keys}
\vspace{-2pt}
\begin{tabular}{@{}lccc@{}}
\toprule
Range Length & Avg Keys & Min--Max & 1-min Terms \\
\midrule
{$<$1 h}   & 6.5  & 1--14 & 30  \\
1--4 h     & 9.2  & 1--20 & 148 \\
4--12 h    & 13.1 & 2--25 & 461 \\
12--24 h   & 15.4 & 4--28 & 961 \\
\bottomrule
\end{tabular}
\end{table}

\subsection{Comparison with Database Indexing}
\label{sec:experiments:pg}

\hspace{1.5em}Compared against PostgreSQL with \texttt{INT4RANGE}+GiST on identical synthetic data, both methods achieve 100\% precision but diverge sharply in scaling: GiST P50 grows $\sim$15$\times$ from 14.7\,ms at 100K to 224\,ms at 1M (P95 up to 360\,ms), while \timehash stays under 12\,ms at 1M; on index size, GiST consumes 119.7\,MB for 1M POIs versus 9.6 terms/doc in an inverted index. Since the two stacks differ (Python in-process vs. PostgreSQL with SQL parsing/IPC), absolute multipliers should not be read as pure algorithmic comparison---the methodologically clean comparison is the within-Elasticsearch evaluation in Section~\ref{sec:experiments:es}. The qualitative point that survives is the architectural mismatch: GiST cannot be embedded within an inverted-index pipeline, requiring a cross-system join that breaks inverted-index composability. Full per-method results are provided in the artifact repository.

\subsection{Within-Elasticsearch Multi-Predicate Composability}
\label{sec:experiments:es}

\hspace{1.5em}Section~\ref{sec:experiments:pg} compared \timehash against PostgreSQL across two distinct stacks, so absolute multipliers cannot be cleanly attributed to the algorithm. To isolate the algorithmic comparison we evaluate two strategies \emph{inside the same Elasticsearch~7.17 deployment} on real data (Yelp Open Dataset, 127{,}123 businesses with hours, categories, and state): (a)~native \texttt{integer\_range} fields with \texttt{contains} queries (BKD trees); and (b)~\timehash keys as ordinary keyword terms with \texttt{terms} filters. \textbf{Both indexes use one document per business with multi-valued range or keyword fields}, eliminating any document-count asymmetry. Both methods compose with non-temporal filters natively via ES \texttt{bool/filter}---no application-layer join is required for either.

We characterize the latency--$K$ curve for three predicate compositions: P1 = time, P2 = time AND category, P3 = time AND category AND state. Categories and states are sampled uniformly from the top five in the dataset. Each measurement uses 1{,}000 random queries (seed 42, 100 warmup, single-shard ES); we report the minimum of 3 timed runs per query under warm cache to suppress JVM GC and OS-level jitter (sub-millisecond $K{=}10$ values are jitter-bounded; the relative ranking is stable across runs). Algorithmic correctness is verified at fetch size 200K: precision\,=\,recall\,=\,1.000 at every $K$ and predicate count for both methods.

\begin{table}[ht]
\centering
\caption{Multi-predicate composability on Yelp Open Dataset (127K businesses, real data; one doc per business, multi-valued fields). P50 latency in ms, 1{,}000 queries (seed~42), 100 warmup queries, ES~7.17 single shard. Precision\,=\,recall\,=\,1.000 verified at all $K$ and predicate counts. \textbf{Bold} = winning method per cell.}
\label{tab:es-kcurve}
\footnotesize
\setlength{\tabcolsep}{3pt}
\begin{tabular}{@{}r rr rr rr@{}}
\toprule
& \multicolumn{2}{c}{\textbf{P1: time}} & \multicolumn{2}{c}{\textbf{P2: +cat}} & \multicolumn{2}{c}{\textbf{P3: +cat+state}} \\
\cmidrule(lr){2-3} \cmidrule(lr){4-5} \cmidrule(lr){6-7}
$K$ & BKD & \timehash & BKD & \timehash & BKD & \timehash \\
\midrule
10      & 0.83 & \textbf{0.38} & 0.89 & \textbf{0.57} & 0.64 & \textbf{0.50} \\
100     & 1.21 & \textbf{0.79} & 1.39 & \textbf{1.04} & 1.27 & \textbf{1.12} \\
1{,}000 & 5.67 & 5.26 & 6.27 & 6.19 & 6.65 & 6.39 \\
10{,}000 & 52.5 & 53.3 & 54.1 & 55.3 & 8.78 & 8.67 \\
\bottomrule
\end{tabular}
\\\vspace{2pt}
\footnotesize
\timehash speedup ($K{=}10$/$K{=}100$): P1 2.17$\times$/1.53$\times$, P2 1.57$\times$/1.33$\times$, P3 1.28$\times$/1.14$\times$.
\end{table}

Table~\ref{tab:es-kcurve} reports the result. \textbf{At production-typical top-$K$ ($K \leq 100$), \timehash achieves a 1.14--2.17$\times$ P50 speedup over BKD} across all three predicate counts. The gap decays as $K$ grows: at $K=1{,}000$ the methods are within 8\%, and at $K=10{,}000$ document materialization dominates and the methods are within 2\% across all predicate counts.

The mechanism is \emph{access-path uniformity}. At low $K$ the query short-circuits after collecting $K$ matches, so the dominant cost is the time to find each match. \timehash's five query keys traverse posting lists and intersect with category/state posting lists entirely within ES's optimized skip-list AND iterator. BKD, in contrast, traverses a tree to produce a doc-id stream that must then compose with posting-list iterators---two access paths instead of one. At large $K$, total query time becomes dominated by \texttt{\_source} fetching and JSON serialization (linear in $K$), hiding the algorithmic gap. Beyond composability, BKD precision is fixed by data type whereas \timehash exposes hierarchy depth as a per-deployment knob (Table~\ref{tab:ablation}); the strongest baseline for \timehash remains inverted-index architectures lacking a BKD-equivalent native range type.

\subsection{Scalability Analysis}
\label{sec:experiments:scale}

\hspace{1.5em}Table~\ref{tab:scalability} validates scalability from 100K to 12.6M synthetic POIs. Three key properties are confirmed: First, average terms per document remains constant at 9.6 regardless of scale, as expected from the algorithm's data-independent decomposition. Second, build time scales linearly with document count: $0.74$\,s at 100K to $87.66$\,s at 12.6M. Third, query latency grows linearly with data size due to Python set-union cost over larger posting lists; a production implementation with roaring bitmaps~\cite{lemire2016roaring} or compressed posting-list pipelines (e.g., PISA~\cite{mallia2019pisa}) would achieve substantially lower latency. At 12.6M POIs, the Python implementation achieves P50 of 116\,ms---still preferable to PostgreSQL GiST's projected multi-second latency at that scale. The sub-12\,ms figure in the abstract refers to \timehash at the 1M scale at which GiST was directly benchmarked (Section~\ref{sec:experiments:pg}); the 116\,ms here reflects Python set-union cost at 12.6M and would drop substantially with a roaring-bitmap-backed production implementation.

\begin{table}[ht]
\centering
\caption{\timehash scalability from 100K to 12.6M synthetic POIs (Python reference implementation, 1,000 queries, seed 42)}
\label{tab:scalability}
\small
\vspace{-2pt}
\begin{tabular}{@{}rcrr r@{}}
\toprule
Scale & Terms/Doc & Build\,(s) & P50\,($\mu$s) & P95\,($\mu$s) \\
\midrule
100K  & 9.6 &  0.74 &     540 &     957 \\
1M    & 9.6 &  8.32 &  11,828 &  62,942 \\
5M    & 9.6 & 48.26 &  45,072 &  73,326 \\
12.6M & 9.6 & 87.66 & 116,421 & 198,427 \\
\bottomrule
\end{tabular}
\end{table}

At 12.6M POIs the index contains ${\sim}121$ million \timehash terms over only 174 unique keys, vs.\ ${\sim}5.6$ billion terms for minute-level indexing---a 46$\times$ reduction. The small unique-key cardinality is what makes high-frequency posting lists (e.g., ``12'' or ``08'') cache-resident in production.

\subsection{Ablation Study}
\label{sec:experiments:ablation}

\hspace{1.5em}Table~\ref{tab:ablation} validates the five-level hierarchy design by systematically removing individual levels and measuring the impact on average key count, computed exhaustively over all possible start/end time combinations at one-minute granularity. Note that these exhaustive averages (5.8 keys/range) are lower than the per-document averages in Table~\ref{tab:index-size} (9.6 terms/doc), because the exhaustive enumeration weights all range lengths equally, whereas real business data is skewed toward longer operating hours (8--12 hours).

\begin{table}[ht]
\centering
\caption{Ablation study: impact of removing hierarchy levels (exhaustive enumeration over all 1-minute-aligned start/end pairs)}
\label{tab:ablation}
\vspace{-2pt}
\begin{tabular}{@{}lcc@{}}
\toprule
Configuration & Avg Keys & $\Delta$ \\
\midrule
Full (4h, 1h, 15m, 5m, 1m) & 5.8  & ---         \\
Remove 4h                   & 8.8  & +51\%       \\
Remove 15m                  & 7.4  & +27\%       \\
Remove 5m                   & 7.5  & +28\%       \\
Remove 1h                   & 6.9  & +19\%       \\
Remove 1m                   & 5.1  & $-$12\%$^*$ \\
\midrule
3-level (4h, 1h, 1m)        & 18.3 & +214\%      \\
4-level (4h, 1h, 15m, 1m)   & 7.5  & +28\%       \\
6-level (+30m)               & 5.5  & $-$5\%      \\
\bottomrule
\multicolumn{3}{@{}l@{}}{\footnotesize $^*$Precision drops to ${\sim}$95\%: non-5-min-aligned boundaries} \\
\multicolumn{3}{@{}l@{}}{\footnotesize cannot be exactly represented without 1-min level.} \\
\end{tabular}
\end{table}

Removing coarse levels (4h, 1h) causes the largest key count increase, as long contiguous ranges can no longer be covered by single keys. Removing the 4-hour level alone increases key count by 51\%, since 8--12 hour businesses that previously needed 2--3 interior keys now need many more 1-hour keys. Removing the 1-minute level reduces key count by 12\% but sacrifices precision, since boundary minutes not aligned to 5-minute multiples cannot be exactly represented (affecting $\sim$0.8\% of POI boundaries). The 3-level alternative (4h, 1h, 1m) produces 214\% more keys due to missing intermediate levels, confirming that intermediate granularities are essential for efficient boundary handling. Adding a 6th level (30m) provides only 5\% improvement over the 5-level design, confirming that our five-level hierarchy achieves the optimal tradeoff between compactness and precision for the business hours domain.

\section{Conclusion}
\label{sec:conclusion}

\hspace{1.5em}We presented \timehash, a hierarchical time-range index materializing a static globally-aligned, domain-informed multi-resolution decomposition as ordinary inverted-index terms. On 100K--12.6M synthetic POIs, \timehash achieves 46$\times$ term reduction at 100\% precision (97.8\%/99.2\% on production/Yelp); within ES with matched indexing, \timehash is 1.14--2.17$\times$ faster than BKD on production top-$K$ multi-predicate workloads ($K{\leq}100$). The decomposition is domain-agnostic---the hierarchy-selection methodology (Section~\ref{sec:experiments:hierarchy}) ports to other domains, and day/date prefixes (e.g., \texttt{mon1212}) handle weekly patterns. Limitations: 24-hour scope, minute-level precision (\texttt{hhmmss} is a direct extension), and five posting-list lookups per query (compensated by ${\sim}240\times$ smaller lists).

\balance

\section*{GenAI Usage Disclosure}

In compliance with the ACM Policy on the use of AI and the CIKM~2026 submission policies, we provide a full disclosure of the use of generative AI tools in the preparation of this work.

\textbf{Tools used.} Anthropic Claude (Sonnet and Opus model families) was used as a writing and code-editing assistant throughout the preparation of this paper.

\textbf{Source material and substantive content.} The problem formulation, algorithm design, theoretical claims, experimental setup, all measured results, and the production deployment experience that motivates this work originate with the authors. Source material provided to the AI assistant included the authors' prior work---a confidential internal engineering tech-share presentation and a published patent disclosure related to the techniques described herein---together with the algorithm specification and the benchmark results.

\textbf{Writing.} Generative AI was used extensively, including to (i)~draft and restructure paragraphs from author-supplied outlines and source material, (ii)~compress and reorganize sections for the page limit, (iii)~suggest framing and positioning relative to related work, and (iv)~edit and proofread author-written prose. Every paragraph in the final manuscript was reviewed, edited where necessary, and accepted by the authors; no technical claim appears in the paper that was not derived from author-supplied source material or author-verified experimental results.

\textbf{Code and experiments.} The C++ reference implementation of the core algorithm (\texttt{getIndexTerms}, \texttt{getQueryTerms}), the correctness invariants, and all experimental measurements reported in this paper were authored, executed, and reproduced by the authors. Generative AI was used to (i)~translate the authors' C++ reference implementation into the additional open-source language ports released alongside this paper (Python, Go, JavaScript, Java, Rust, Kotlin, Scala); (ii)~draft benchmark harnesses, plotting scripts, and parser utilities; and (iii)~suggest Python and Elasticsearch API usage patterns. All such code was executed, tested, and verified by the authors against the C++ reference. The seeded benchmark scripts released in the artifact repository byte-reproduce all numbers reported here.

\textbf{Data.} No generative AI was used to synthesize or modify any data reported in this paper. The synthetic POI generator (Section~\ref{sec:experiments:setup}) is a deterministic procedure parameterized by production distributions characterized in Section~\ref{sec:problem}; the Yelp Open Dataset is used as released by Yelp under its terms.

\textbf{AI-assisted sections (specific scope).} For transparency, we list below the sections in which AI substantively assisted with drafting or restructuring from author-supplied source material and benchmark results:
\begin{itemize}
  \item \S\ref{sec:related} Related Work: compression from four subsections to three; merging of database and spatial-hierarchy discussions.
  \item \S\ref{sec:problem} Problem Definition: redraft from four subsections into a single paragraph.
  \item \S\ref{sec:algorithm} \timehash Algorithm: compression of five subsections to four; removal of one redundant figure and tightening of the worked example.
  \item \S\ref{sec:analysis} Theoretical Analysis: redraft from theorem-form proofs into three paragraphs.
  \item \S\ref{sec:implementation} Implementation: compression of two subsections into one paragraph.
  \item \S\ref{sec:experiments:setup} Experimental Setup opening; \S\ref{sec:experiments:pg}/\S\ref{sec:experiments:es} prose refinements.
  \item \S\ref{sec:intro:contrib} Contributions list: restructuring from six items to four.
  \item \S\ref{sec:conclusion} Conclusion: merger with the prior Discussion section.
  \item Bibliography: AI suggested candidate references missing from earlier drafts (BKD original~\cite{procopiuc2003bkd}, Roaring bitmaps~\cite{lemire2016roaring}, TB-tree~\cite{pfoser2000tbtree}, Druid~\cite{yang2014druid}, PISA~\cite{mallia2019pisa}), drafted the corresponding \texttt{bibitem} entries, and cross-checked their metadata; the authors independently verified each reference's existence and metadata via Google Scholar / DBLP and confirmed that the citation supports the claim at its body-text location.
  \item This GenAI Usage Disclosure section itself.
\end{itemize}

\textbf{Representative prompt categories.} Prompts fell into categories such as: (i)~``compress this section to $N$ words while keeping these key claims''; (ii)~``merge subsections X and Y, prioritizing the technical content''; (iii)~``rewrite this paragraph for clarity, do not change any numbers''; (iv)~``translate the C++ reference implementation in file X.cpp to language Y, preserving the public API''; (v)~``suggest a more compact wording for [phrase] under [$N$-character] limit''. No prompts requested the generation of novel algorithmic ideas, novel experimental results, or interpretations not derivable from author-supplied source.

\textbf{Post-editing process.} For each AI-produced draft, the authors verified every quantitative claim against the source material or experimental logs, edited prose for accuracy and tone, removed any unsupported statements, and confirmed that all citations and \LaTeX{} cross-references resolve correctly. Rebuilding the manuscript was performed after each substantive change to catch broken references and to monitor the page budget. The worst-case key-count computation in Section~\ref{sec:analysis} ($B{=}24$, 30-key upper bound, 28-key empirical worst) was independently re-derived by the authors; reference metadata was re-checked against Google Scholar / DBLP entries.

\textbf{Tasks not AI-assisted.} The C++ reference implementation, the algorithm design, the data-driven hierarchy selection (including the enumeration in Table~\ref{tab:hierarchy-opt}), all experimental measurements (latency, term counts, precision/recall, ablation), the production deployment decisions described in this work, the patent invention disclosure, and the choice of baselines and evaluation methodology were authored entirely by the authors without AI assistance.

\textbf{Accountability.} The authors take full responsibility for the accuracy, originality, and integrity of all content in this paper.

\bibliographystyle{ACM-Reference-Format}

\end{document}